\documentclass[format=acmsmall, review=false]{acmart}
\usepackage{acm-ec-24}
\usepackage{booktabs} 
\setcitestyle{acmnumeric}
\usepackage{cleveref}
\usepackage{tcolorbox}
\usepackage{threeparttable}
\usepackage{tabularx, booktabs} 

 \usepackage{pdfpages}
 
\newcommand{\bal}{\ensuremath{\mathsf{BALANCE}}\xspace} 
\newcommand{\val}{\mathsf{Val}}

\newcommand{\du}{\ensuremath{\dot{u}}\xspace}

\usepackage{macros_pan}
\newcommand{\anhai}[1]{{\color{black} #1}}

\title{A  Variational-Calculus Approach to Online Algorithm Design and Analysis}
 \titlenote{This manuscript is a work in progress, and feedback is welcome. Please send inquiries and comments to \bluee{pxu@njit.edu}.}
\author{Pan Xu}

\authornote{Department of Computer Science,  
New Jersey Institute of Technology, Newark, NJ 07102, \texttt{pxu@njit.edu}}

\begin{abstract}
Factor-revealing linear programs (LPs) and policy-revealing LPs arise in various contexts of algorithm design and analysis. They are commonly used techniques for analyzing the performance of approximation and online algorithms, especially when direct performance evaluation is challenging. The main idea is to characterize the worst-case performance as a family of LPs parameterized by an integer \( n \ge 1 \), representing the size of the input instance. To obtain the best possible bounds on the target ratio (e.g., approximation or competitive ratios), we often need to determine the optimal objective value (and the corresponding optimal solution) of a family of LPs as \( n \to \infty \). One common method, called the \emph{Primal-Dual} approach, involves examining the constraint structure in the primal and dual programs, then developing feasible analytical solutions to both that achieve equal or nearly equal objective values. Another approach, known as \emph{strongly factor-revealing LPs}, similarly requires careful investigation of the constraint structure in the primal program, followed by devising new constraints to replace existing ones. In summary, both methods rely on \emph{instance-specific techniques}, which is difficult to generalize from one instance to another.

In this paper, we introduce a general variational-calculus approach that enables us to analytically study the optimal value and solution to a family of LPs as their size approaches infinity. The main idea is to first reformulate the LP in the limit, as its size grows infinitely large, as a variational-calculus instance and then apply existing methods, such as the Euler-Lagrange equation and Lagrange multipliers, to solve it. We demonstrate the power of our approach through three case studies of online optimization problems and anticipate broader applications of this method. We hope our work can serve as an introductory example of using continuous methodologies to address algorithm design and analysis for discrete optimization problems.
\end{abstract}

\begin{document}
\begin{titlepage}
\maketitle
\end{titlepage}

\newpage

\section{Introduction}
Variational calculus is a well-established field of mathematical analysis that aims to find an optimal function (or curve) that optimizes a given functional. Unlike traditional optimization, which seeks a specific point that optimizes a target function, variational calculus is designed to optimize an entire function (or curve) within a given function space. It has applications across various science and engineering disciplines. For example, in image processing, variational methods are used to remove noise and enhance clarity by minimizing energy functionals. In control theory, variational techniques help design an optimal control function to achieve a desired goal; for more examples, see the book~\cite{Cassel_2013}. A generic variational instance takes \anhai{the} form as follows:
\begin{align}\label{vc:ge}
\min_{u(t): t\in [a,b]} J[u]:=\int_a^b L(t,u,\du)~\sd t: ~~u(a)=\alp, u(b)=\beta,
\end{align}
where $t$ is called \emph{independent} variable, and $u$ is called \emph{dependent} variable that is assumed to be a function of $t \in [a,b]$, $L$ is a function of $t$, $u$ and $\du$,\footnote{Throughout this paper, we use $\du$ and $u'$ interchangeably, both denoting   the first derivative of $u$ with respect to $t$.} called Lagrangian function, which is assumed to have almost continuous second-order partial derivatives, and $J[u]$ is called objective functional. Putting into plain language, the variational instance \eqref{vc:ge} aims to identify a function $u(t)$ with $t \in [a,b]$ that satisfies boundary conditions $u(a)=\alp, u(b)=\beta$ and minimizes the functional $J(u)$.

\xhdr{Extensions of Variational Instances}. There is a \emph{vast} number of extensions to the generic form shown in~\eqref{vc:ge}. We present a few common examples as follows. (1) Second-Order (or higher-order) variational instances. In this case, the Lagrangian function typically takes the form \( L(t, u, \du, \ddot{u}) \), involving the second derivative of \( u \) with respect to the independent variable \( t \). (2) Variational instances with multiple dependent variables. Here, the Lagrangian function may take a form such as \( L(t, u, v, \du, \dot{v}) \), where both \( u \) and \( v \) are functions of \( t \). (3) Multi-dimensional variational instances. Consider a two-dimensional case, for example. The objective functional typically takes the form:
\[
J[u] := \iint_{\Omega} L(t, s, u, \partial u / \partial t, \partial u / \partial s) \, \sd s \, \sd t,
\]
where there are two independent variables, \( t \) and \( s \), and one dependent variable \( u \), assumed to be a function of both \( t \) and \( s \). (4) Non-boundary constraints. Consider the instance~\eqref{vc:ge} of the generic form as an example. Examples of such constraints include (i) integral constraints of the form \( \int_a^b F(t, u, u') \, \sd t = 0 \); (ii) holonomic constraints of the form \( G(t, u) = 0 \); and (iii) non-holonomic constraints of the form \( K(t, u, u') = 0 \). \emph{For each extension, numerous techniques have been developed to solve it; see~\cite{dacorogna2024introduction, olver2021calculus} for details.}

\smallskip

As a powerful mathematical tool, variational calculus has, however, seen rare applications in theoretical computer science, particularly in areas like algorithm design and analysis. One primary reason is that most optimization problems in computer science, such as NP-hard problems, involve \emph{discrete} decision variables and objectives. In contrast, variational calculus techniques are generally applied to find optimal \emph{continuous} functions within a given function space. Together, these fundamental differences help explain why variational calculus methods have limited applications in theoretical computer science.

In this paper, we address the gap described above by presenting examples that demonstrate how variational techniques can be applied to analyze (discrete) algorithms for online optimization problems. Specifically, we consider three representative cases: Online Matching (under adversarial order), Adwords, and the Secretary Problem, detailed as follows.
\subsection{Online Matching, Online $b$-Matching, and Secretary Problem} 
\xhdr{Online Matching}. Online (Bipartite) Matching was first introduced by~\cite{kvv}. In this problem, we have an (unweighted) bipartite graph $G=(U,V,E)$, which is unknown to the algorithm (\alg). The vertices of $U$ are static, while the vertices of $V$ are revealed (called \emph{arrive}) sequentially in an adversary order (pre-arranged by an adversary oblivious to the algorithm). Upon the arrival of a vertex $v$, \alg gets informed of all $v$'s neighbors, and it needs to match $v$ to one \anhai{of} its available neighbor $u$, if any, before observing any future arrivals from $V$, and the decision is irrevocable. By default, each $u \in U$ is assumed to have a unit matching capacity. The goal is to design an online matching algorithm \alg that maximizes the  expected total number of matches. Since the introduction of the aforementioned basic version, numerous related variants have been proposed and studied (\cite{mehta2012online}), including the two important models below.

\xhdr{Adwords}. As an important variant, \anhai{the} Adwords problem strictly generalizes Online Matching, bringing it closer to real-world Internet advertising. First, each vertex $u \in U$ is associated with a given budget $b_u>0$. Second, upon the arrival of an online vertex $v \in V$, an algorithm \alg gets informed not only \anhai{of} the set $N_v \subseteq U$ of $v$'s neighbors but also of the edge weight $b_{u,v}>0$     for each $u \in N_v$. This is called the \emph{bid} of $v$ for $u$. Specifically, suppose at the time of $v$'s arrival, each $u \in  N_v$ has a real-time budget $\tilde{b}_u$, which initially equals $b_u$ but decreases over time. The algorithm \alg can assign $v$ to any neighbor $u \in N_v$, gaining a reward of $\min \sbp{b_{u,v}, \tilde{b}_u}$ and having $u$'s budget reduced by the same amount. The goal is updated  to maximize the total expected rewards. We can verify that the aforementioned Online Matching can be cast as a special case of Adwords when all $u$ have a unit budget and all bids are equal to one.  

\xhdr{Online $b$-Matching}.  The online $b$-matching problems is a special case of the Adwords problem when all $u$ have a budget $b_u=b \in \mathbb{Z}_{\ge 1}$ and all bids are equal to one, with $b$ being a parameter accessible to the algorithm. Alternatively, it can \anhai{be} viewed as a generalized version of Online Matching such that each $u$ has a uniform matching capacity of $b$.

\xhdr{Secretary Problem}. The Secretary Problem, also known as the ``optimal stopping problem,'' is a classical problem in decision theory, Operations Research, and Computer Science. The basic setting of the problem is as follows: Suppose we need to hire one secretary from a pool of $n$ candidates who arrive in a random order. The challenge is that before interviewing the next candidate, we must make an irreversible decision to either accept the current one (and stop permanently) or reject it. Our goal is to maximize the probability of selecting the best candidate based on a sequence of decisions.  For this basic setting, a well-known optimal $1/\sfe$-competitive strategy is to first interview a batch of $n/\sfe$ candidates, reject all of them, and then accept the first candidate who beats the best candidate from the initial batch.

Since the introduction of the Secretary Problem, numerous generalizations of the basic setting have been proposed; see, \eg~\cite{oveis2013variants,kleinberg2005multiple} Among the various variants of the Secretary Problem, notable examples include the $J$-choice, $K$-best secretary problem, also referred to as $(J,K)$-secretary problem, which was introduced by~\cite{10.1287/moor.2013.0604}. In this variant, the algorithm is allowed to select up to $J$ candidates, aiming to choose as many as $K$ top-ranked candidates. Additionally, ~\citep{10.1287/moor.2013.0604} proposed a powerful linear programming (LP)-based technique that aids in designing and analyzing algorithms for the Secretary Problem and its various variants; see the discussion of the ``Policy-Revealing LP'' in the paragraph below.

\subsection{Factor-Revealing and Policy-Revealing LPs: Limitations of Existing Approaches}

\subsubsection{Factor-Revealing LPs and Policy-Revealing LPs}
The factor-revealing linear program (LP) is a technique commonly used to analyze the performance of approximation and online algorithms. It is particularly effective for algorithms where the approximation ratio or competitive ratio is challenging to compute directly. Examples include online (bipartite) matching~\cite{mehta2007adwords} and facility location problems~\cite{mahdian2006approximation}. This approach formulates the task of identifying the worst-case performance of a target algorithm as the optimization of a family of linear programs parameterized by an integer \( n \ge 1 \), representing the LP size such as the number of variables. For maximization problems such as online matching, a family of factor-revealing LPs is constructed to satisfy the following properties: (1) The optimal value of the LP is non-increasing as \( n = 1, 2, \ldots \); (2) The infimum of the optimal values as \( n \to \infty \) serves as a valid lower bound on the resulting approximation or competitive ratio of a target algorithm. \emph{Note that the optimal value of an LP parameterized by any specific integer \( n \) may not necessarily provide a valid lower bound on the target ratio. This makes it necessary to compute the asymptotic optimal value of the LP as its size \( n \to \infty \).
}

The policy-revealing LP was pioneered by~\cite{10.1287/moor.2013.0604} in their study of the classical Secretary Problem and its variants. Compared to factor-revealing LPs, the policy-revealing LP exhibits more power: this approach formulates the task of finding \emph{the best algorithm} as the optimization of a family of linear programs parameterized by an integer $n \ge 1$, representing the total number of candidates.

\subsubsection{Existing Approaches and Limitations}
For both Factor-Revealing LPs and Policy-Revealing LPs, we need to solve a family of LPs parameterized by an integer \( n \ge 1 \), representing the size of the LP. Due to their nature, we are particularly interested in determining the optimal values and solutions as the LP size \( n \to \infty \), which poses significant technical challenges. 

\xhdr{Primal-Dual Method}. A common approach to tackling the challenge above is the primal-dual method: By carefully examining the constraint patterns in the primal LP and its dual, we construct an \anhai{analytically} feasible solution for each program as a function of $n$ \anhai{and} then prove the equality or near equality of the objective values associated with these solutions.  The limitations of this approach are evident. As noted in~\cite{mahdian2011online}, the process of constructing \anhai{analytically} feasible solutions for both the primal and the dual programs with equal or nearly equal objective values ``can be painstaking and is often doomed to incur a constant loss in the factor (see~\cite{mahdian2006approximation,jain2003greedy} for an example).''

\xhdr{Strongly Factor-Revealing LPs}.~\cite{mahdian2011online} invented another elegant and powerful technique, called \emph{strongly factor-revealing LPs}, to analyze the performance of \rank  for online matching under random arrival order based on a family of factor-revealing LPs.  The main idea is to replace the existing constraints of factor-revealing LPs with newly developed ones so that each resulting LP has an optimal value that provides a valid lower bound on the competitiveness of \rank, regardless of its size \( n \). Consequently, this method allows us to bypass computing the asymptotic optimal objective value of the factor-revealing LPs as \( n \to \infty \). Instead, it suffices to compute the optimal value of any LP with the newly constructed constraints for any size of \( n \) to obtain a valid lower bound on the competitiveness of \rank.\footnote{As expected, the strongly factor-revealing LPs developed in~\cite{mahdian2011online} have increasing optimal values for \( n = 1, 2, \ldots\). This allows us to leverage computational power to determine the optimal value of an LP with size $n$ as large as possible, yielding the best possible lower bound on the target ratio.} 

There are several limitations associated with the technique of strongly factor-revealing LPs. First, it provides only a lower bound on the competitiveness or approximation ratio of a target algorithm for maximization problems. Additionally, the quality of the “best” achievable lower bound depends on both the available computational power and the convergence speed of the optimal values of strongly factor-revealing LPs as the size \( n \to \infty \). Second, similar to the primal-dual method, we must carefully analyze the structure of constraints in the factor-revealing LPs to develop new constraints, ensuring that the resulting LPs qualify as strongly factor-revealing LPs. This process demands considerable technical effort and numerous intricate tricks.

{To summarize, both methods above require \emph{instance-specific tricks and ideas}, which are difficult to generalize from one instance to another.}

\subsection{Our Approach: A Variational-Calculus Method}

We introduce a \emph{general} variational-calculus approach to analytically study the optimal objective value and solution to a family of LPs as the size \( n \to \infty \). The main idea is as follows: we first reformulate the problem in the limit, as the LP size grows infinitely large, as a variational-calculus instance and then apply existing methods to solve it. As mentioned earlier, since variational calculus is a well-established field, there exists a rich toolkit of techniques we can use to tackle even complex variational instances. Below, we present an illustrative example to demonstrate how our method works.

\begin{example}[An Illustrative Example of the Variational-Calculus Approach]\label{exam:29-a}
Consider the toy example of  factor-revealing LPs presented in~\cite{mahdian2011online} as follows:\footnote{We use slightly different notation here to maintain consistency with the rest of the paper.}
\begin{align}\label{lp:29-a}
\min_{\x=(x_i) \in [0,1]^n} \frac{1}{n} \sum_{i=1}^n x_i:~~ 1-x_i \le \frac{1}{n} \sum_{1\le \ell<i } x_\ell, \forall i \in [n]; x_i \ge x_{i+1}, \forall 1\le i<n.
\end{align}
\end{example}
Let \( \tau(n) \) denote the optimal objective value of the LP above. As noted by~\cite{mahdian2011online}, we can apply either the primal-dual method or the strongly factor-revealing LP technique to show that \( \tau^* := \lim_{n \to \infty} \tau(n) = 1 - 1/\sfe \).

For each \( i \in [n] \), let \( g(i/n) := x_i \), where \( g(t) \) is a continuous function with both its domain and range on \([0,1]\). We can prove that \LP-\eqref{lp:29-a} as its size $n \to \infty$ becomes equivalent to the following variational instance:

\begin{align}\label{vc:29-b}
\min_{g(t): t \in [0,1]} \int_0^1 g(t) \sd t:~~ 1-g(t) \le \int_0^t g(z) \sd z, \forall t \in [0,1]; 1\ge g(t) \ge g(t') \ge 0, \forall 0 \le t< t' \le 1.
\end{align}

For factor-revealing LPs, we are primarily concerned with the optimal objective values of \LP-\eqref{lp:29-a} as its size $n \to \infty$, i.e., the value \( \tau^* \), rather than the corresponding optimal solution. The equivalence between the two programs suggests that we can reduce the task of determining \( \tau^* \) to finding the optimal value of the variational instance~\eqref{vc:29-b}. In this context (when only the optimal value of a variational instance is of interest), we do not necessarily need to invoke standard variational techniques to obtain an optimal solution. Instead, simpler alternatives, such as an Ordinary Differential Equation (ODE)-based approach, may suffice. Consider the specific instance~\eqref{vc:29-b}. By setting \( u(t) := \int_0^t g(z) \, \mathrm{d}z \), we can reformulate instance~\eqref{vc:29-b} as follows:
\begin{align}\label{vc:29-d}
\min_{u(t): t \in [0,1]} u(1):~~u(t) + \du(t) \ge 1, \forall t \in [0,1];  u(0) = 0; 1 \ge \du(t) \ge \du(t') \ge 0,  \forall 0 \le t < t' \le 1.
\end{align}
For the instance~\eqref{vc:29-d}, a simple trick (which works most of the time, though not necessarily always) is to set the main constraint tight, i.e., \( u(t) + \du(t) = 1 \). This equality, combined with the initial condition \( u(0) = 0 \), yields a unique solution of \( u(t) = 1 - \sfe^{-t} \), which implies that \( u(1) = 1 - 1/\sfe \). Indeed, we can verify that this is the optimal objective value. We demonstrate this approach on two case studies of factor-revealing LPs and provide \anhai{rigorous} proof justifying it; see Section~\ref{sec:bal} and Appendix~\ref{app:rank}.

For policy-revealing LPs, we are concerned with both the optimal value and the optimal solution. For simple variational instances, the ODE-based technique described above suffices to obtain an optimal function. Consider the specific instance~\eqref{lp:29-a} as an example of a \emph{policy-revealing LP}; we can verify that the ODE-based approach yields a {unique} optimal function \( g(t) = \sfe^{-t} \) for the reformulated instance~\eqref{vc:29-b}. For more complex instances, however, we have to invoke standard variational techniques, which can yield not only an optimal solution but also potentially address the issue of uniqueness; see the case study in Section~\ref{sec:sp}.

\subsection{Main Contributions}

In this paper, we present a general Variational-Calculus (VC) approach as an alternative method for solving factor-revealing and policy-revealing LPs that arise in various contexts of online algorithm design and analysis. Specifically, we present three case studies to demonstrate the applicability of our approach. 

\subsubsection{Resolving the Factor-Revealing LP of \bal due to~\cite{mehta2007adwords} via a VC Approach} 

The algorithm \bal was first introduced by~\cite{kalyanasundaram2000optimal} for the Online $b$-Matching problem, which proves that $\bal$ achieves \anhai{optimal} competitiveness of $1-1/\sfe$ when $b \to \infty$. The work~\cite{mehta2007adwords} re-analyzed $\bal$ in a broader context of the Adwords problem, a generalized  version of Online $b$-Matching. Specifically, they presented an elegant combinatorial analysis for \bal and characterize the competitiveness of $\bal$ using a factor-revealing linear program, as shown in \LP-\eqref{lp:bal}. 

\begin{align}\label{lp:bal} 
\max_{\x=(x_i) \in [0,1]^N} \sum_{i \in [N]} x_i \cdot \bp{1-\frac{i}{N}}: ~~
 \sum_{1 \le i \le p} x_i \bp{1+\frac{p-i}{N}} \le  \frac{p}{N},  \forall p \in [N].
\end{align}

The following theorem is implied by the work~\cite{mehta2007adwords}.
\begin{theorem}[\cite{mehta2007adwords}]\label{thm:re-bal}
Let $\LP_B(N)$ denote \LP-\eqref{lp:bal} parameterized with an integer $N$. For Online $b$-Matching, $\bal$ achieves a competitiveness of at least $1-\val(\LP_B(\infty))$ when $b \to \infty$, where $\val(\LP_B(\infty))$ denotes the optimal value of $\LP_B(N)$ when $N \to \infty$.
\end{theorem}

We emphasize that \emph{neither  \LP-\eqref{lp:bal} nor Theorem~\ref{thm:re-bal} is explicitly stated in~\cite{mehta2007adwords}}.  For completeness, we  provide  a formal proof of Theorem~\ref{thm:re-bal} in Appendix~\ref{app:bal}.\footnote{While following  the main idea in~\cite{mehta2007adwords}, we have added several necessary modifications to make the proof of Theorem~\ref{thm:re-bal} more rigorous compared with the version presented in~\cite{mehta2007adwords}.} As suggested by Theorem~\ref{thm:re-bal}, the task of lower bounding the competitiveness of $\bal$ is reduced to that of computing the value of $\val(\LP_B(\infty))$. Instead of manually identifying a pair of optimal solutions to the primal program of $\LP_B(N)$ in~\eqref{lp:bal} and its dual for any given $N \in \mathbb{Z}_{\ge 1}$, as did in~\cite{mehta2007adwords}, we present a variational-calculus  approach to exactly solving for the value $\val(\LP_B(\infty))$. Specifically, we propose a variational instance as follows:
 \begin{align}\label{vc:bal}
\max_{g(t): t \in [0,1]} & \int_0^1 g(t) (1-t) \sd t:~~ \int_0^t  g(z) \sbp{1-z+t} \sd z \le t,~~ \forall t \in [0,1].
\end{align}

Additionally, we prove that 
\begin{theorem}[Section~\ref{sec:bal}]\label{thm:bal}
When $N \to \infty$, $\LP_B(N)$ in \LP-\eqref{lp:bal} becomes equivalent to the variational instance~\eqref{vc:bal}, which has an optimal objective value of $1/\sfe$.
\end{theorem}

According to Theorem~\ref{thm:re-bal}~\cite{mehta2007adwords}, the competitiveness achieved by \bal is lower bounded by $1-\val(\LP_B(\infty))$, which is equal to $1-1/\sfe$.  This aligns with the findings in~\cite{kalyanasundaram2000optimal,mehta2007adwords}.

\subsubsection{Resolving the Factor-Revealing LP of \rank due to~\cite{mehta2012online} via a VC Approach}  The classical algorithm \rank was first introduced by~\cite{kvv}, which proved that \rank achieves an optimal competitiveness of $1-1/\sfe$ for the Online Matching problem (under adversarial order).  The book~\cite{mehta2012online} (Chapter 3) offers a refined combinatorial analysis for \rank. Specifically, it proposed a factor-revealing linear program (LP), as shown in LP-\eqref{lp:rank}, whose optimal value provides a valid lower bound on the competitiveness achieved by \rank when the number of variables $n$ (representing the input graph size) approaches infinity.\footnote{For completeness, we provide details of the \rank algorithm and how~\cite{mehta2012online} derived LP-\eqref{lp:rank} in Appendix~\ref{app:rank}.}
\begin{align}\label{lp:rank}
{\min_{\x =(x_i)\in [0,1]^n} \frac{1}{n}\sum_{i\in [n]} x_i:~~ x_i+\frac{1}{n} \sum_{1\le j \le i} x_j \ge 1, \forall i \in [n].}
\end{align}

Let $\LP_R(n)$  denote  \LP-\eqref{lp:rank} parameterized with an integer $n \ge 1$. We propose a variational-calculus instance as follows:
\begin{align}\label{vc:rank}
{\min_{g(t): t\in [0,1]} \int_0^1 g(t)~ \sd t : g(t)+\int_0^t g(z) ~\sd z \ge 1, \forall t \in [0,1].}
\end{align}

Additionally, we prove that
\begin{theorem}[Appendix~\ref{app:rank}]\label{thm:rank}
When $n \to \infty$, $\LP_R(n)$ in \LP-\eqref{lp:bal} becomes equivalent to the variational instance~\eqref{vc:rank}, which has an optimal objective value of $1/\sfe$.
\end{theorem}

Our result is consistent with that in~\cite{kvv,mehta2012online}. The proof is quite similar to that of Theorem~\ref{thm:bal} and we defer it to Appendix~\ref{app:rank}.


\subsubsection{Resolving the Policy-Revealing LP for the Classical Secretary Problem due to~\cite{10.1287/moor.2013.0604} via a VC Approach} 
The work~\cite{10.1287/moor.2013.0604} proposed different variants of the Secretary Problem, and for each variant, it presented a policy-revealing LP parameterized by an integer $n$, representing the total number of candidates accessible to the algorithm. Among all the LPs presented, the policy-revealing LP for the classical secretary problem, as shown in LP-\eqref{lp:sp}, serves as the \emph{foundational blueprint} on which all others are based.

\begin{align}\label{lp:sp}
{\max_{\x =(x_i)\in [0,1]^n} \sum_{i\in [n]} x_i \cdot (i/n):~~ x_i \cdot i \le 1-\sum_{1 \le \ell<i} x_\ell, \forall i \in [n].}
\end{align}
 
According to~\cite{10.1287/moor.2013.0604}, the variable $x_i \in [0,1]$ in LP-\eqref{lp:sp} represents the overall probability that any policy $\alg$ selects the candidate at position $i \in [n]$ (\ie arriving at time $i$), and the objective captures the overall probability that $\alg$ selects the best candidate. Note that the probability $x_i$ accounts for the following three random events: $\alg$ reaches position $i$ (\ie it does not stop earlier), the candidate at position $i$ is the best among all so far,\footnote{Recall that for the classical secretary problem, our goal is to maximize the probability of selecting the best candidate. As a result, we can safely restrict our attention to policies that select only the best-so-far candidate.} and $\alg$ opts to accept the candidate at position $i$. The work~\cite{10.1287/moor.2013.0604} proves the equivalence between a feasible policy for the classical secretary problem and a feasible solution to LP-\eqref{lp:sp}, as stated in the theorem below.

\begin{theorem}[\cite{10.1287/moor.2013.0604}]\label{thm:re-sp}
Any feasible solution $\x=(x_i)$ to LP-\eqref{lp:sp} one-one corresponds to a viable policy $\alg(\x)$ for the classical secretary problem, such that $\alg(\x)$ picks the best candidate with probability equal to the objective value on $\x$, which is $\sum_{i \in [n]} x_i \cdot (i/n)$.
\end{theorem}

Theorem~\ref{thm:re-sp} suggests that we can reduce the task of identifying an optimal policy for the classical secretary problem  to that of finding an optimal solution to LP-\eqref{lp:sp}.\footnote{As pointed out in~\cite{10.1287/moor.2013.0604}, we can recover the policy $\alg(\x)$ from any feasible solution $\x$ simply as follows: Upon the arrival of candidate at position $i \in [n]$, the policy should accept the candidate  with probability $x_i \cdot i/(1-\sum_{1 \le \ell<i} x_\ell)$ provided that the policy has not stopped earlier and  candidate $i$ has the best performance so far; otherwise, the policy should reject the candidate.} 
To solve an optimal solution for the LP-\eqref{lp:sp}, the work~\cite{10.1287/moor.2013.0604} follows a common practice by first guessing a pair of primal and dual solutions and then proving each feasibility and the equality of the two objective values.  In contrast, we propose a variational instance as follows:
 \begin{align}\label{vc:sp}
\max_{g(t): t \in [0,1]} & \int_0^1 g(t)  \sd t:~~0 \le g(t) \le 1- \int_0^t  \frac{g(z)}{z}  \sd z,~~ \forall t \in [0,1].
\end{align}

Let $\LP_S(n)$  denote  \LP-\eqref{lp:sp} parameterized with an integer $n \ge 1$.  We prove that
\begin{theorem}[Section~\ref{sec:sp}]\label{thm:sp}
(1) When \( n \to \infty \), \( \LP_S(n) \) in \LP-\eqref{lp:sp} becomes equivalent to the variational instance~\eqref{vc:sp}. (2) The instance~\eqref{vc:sp} admits a unique maximizer \( g^*(t) \) with an optimal objective value of \( 1/\sfe \), where \( g^*(t) \) takes the form \( g^*(t) = 0 \) for \( t \in [0, 1/\sfe] \) and \( g^*(t) = (1/\sfe)/t \) for \( t \in (1/\sfe, 1] \).\footnote{More precisely, the uniqueness of the maximizer \( g^* \) holds up to a set of measure zero. In other words, any other maximizer for instance~\eqref{vc:sp} must be equal to \( g^* \) almost everywhere, except possibly on a set of measure zero.}
\end{theorem}

\xhdr{Remarks on Theorem~\ref{thm:sp}}. (i) As suggested in the proof of Claim (1) of Theorem~\ref{thm:sp} in Section~\ref{sec:sp-a}, the value $g(t)$ for any $t \in [0,1]$ can be viewed as an approximation of $x_{i} \cdot i$ with $i=\lfloor t \cdot n \rfloor$ when $n \gg 1$. As a result, the  optimal policy  
corresponding to $g^*$ should accept any candidate at position $i$ with probability equal to 
\[
\frac{x^*_i \cdot i}{1-\sum_{1 \le \ell<i} x^*_\ell} \approx \frac{g^*(i/n)}{1-\int_0^{i/n}\frac{g^*(z)}{z} \sd z}=\begin{cases}
0, & \text{if $i/n \le  1/\sfe$},\\
1, &\text{if $i/n>1/\sfe$},
\end{cases}
\]
given that the optimal policy has not stopped before $i$ and candidate $i$ performs best so far. This recovers the result in~\cite{10.1287/moor.2013.0604}. (ii) The uniqueness of the maximizer for the instance~\eqref{vc:sp} implies the \anhai{uniqueness} of the optimal policy for the classical Secretary Problem, due to the result of Theorem~\ref{thm:re-sp}. This highlights the advantages of the variational-calculus-based approach over the traditional one: the former can not only identify an optimal solution but also address the issue of uniqueness that is beyond the scope of the latter. (iii) Note that the policy-revealing \LP-\eqref{lp:sp} formulated for the classical Secretary Problem serves as the foundation for the policy-revealing LPs used across other Secretary Problem variants proposed in~\cite{10.1287/moor.2013.0604}. This foundational role supports our confidence that the variational-calculus approach can be effectively extended to these other cases as well.

\subsection{Other Related Work}

Our work is part of a long line of studies that use continuous methodologies to address algorithm design and analysis for discrete optimization problems, particularly in online optimization. We present a few recent examples as follows: \cite{huang2021online} introduced a concept called the \emph{Poisson Arrival Model} as an alternative to the \emph{Known Independent Identical Distributions} (KIID) model. They proved the equivalence between these two models in competitive analyses for online algorithms, showing that the proposed Poisson arrival model makes it much easier to apply continuous functions to assist in algorithm design and analysis compared with its discrete counterpart (KIID). In their subsequent study~\cite{huang2022power}, they considered edge-weighted online matching with free disposal and successfully utilized a second-order ordinary differential inequality to characterize the entire progress of the top-half sampling algorithm. More recently,~\cite{macrury2024random} introduced an approach for designing Random-Order Contention Resolution Schemes via any exact solution in continuous time.

\section{Resolving the Factor-Revealing LPs of $\bal$ via a Variational-Calculus Approach}\label{sec:bal}
Recall that $\LP_B(N)$ denotes the factor-revealing LP for the Online $b$-Matching problem, as defined in~\LP-\eqref{lp:bal}, which is parameterized with an integer $N$ with $N \gg 1$.  

\subsection{Formulating $\LP_B(N)$ as a Variational-Calculus  Instance}\label{sec:eq-b} 
\begin{lemma}\label{lem:10-16-a}
When $N \to \infty$, $\LP_B(N)$ defined in~\LP-\eqref{lp:bal} becomes equivalent to the variational-calculus instance shown in~\eqref{vc:bal}.
\end{lemma}

\begin{proof}
For each $i \in [N]$, let $g(i/N)=N \cdot x_i$, where $g(t)$ is a  function over $t \in [0,1]$. Observe that for the objective of $\LP_B(N)$ as shown in~\LP-\eqref{lp:bal}:
\begin{align*}
 \sum_{i \in [N]} x_i \cdot \bp{1-\frac{i}{N}} &=\frac{1}{N}\sum_{i \in [N]} g(i/N) \cdot \bp{1-\frac{i}{N}}  \longrightarrow  \int_0^1 g(t) (1-t) \sd t,~~\sbp{\mbox{when $N \to \infty$}}.
\end{align*}
As for the constraint of $\LP_B(N)$  in~\LP-\eqref{lp:bal},
\begin{align*}
 &\sum_{1 \le i \le p} x_i \cdot \bp{1+\frac{p-i}{N}} \le  \frac{p}{N},  \forall p \in [N]  \Leftrightarrow
 \sum_{1 \le i \le p} x_i \cdot\bp{1+\frac{p-i}{N}} \le  \frac{p}{N},  p=\lfloor t \cdot N \rfloor, \forall t \in [0,1].\\
  \Leftrightarrow&
 \frac{1}{N}\sum_{1 \le i \le \lfloor t \cdot N \rfloor} g(i/N) \cdot \bp{1+\lfloor t \cdot N \rfloor/N-i/N} \le  \lfloor t \cdot N \rfloor/N, \forall t \in [0,1]. \\
   \Leftrightarrow&  \int_0^t  g(z) \sbp{1+t-z} \sd z \le t,~~  \forall t \in [0,1]~~\sbp{\mbox{when $N \to \infty$}}.
 \end{align*}
 Wrapping up the analysis above, we establish the claim.
 \end{proof}

\subsection{Solving the Variational-Calculus Instance~\eqref{vc:bal} Optimally}\label{sec:29-b} 
We can apply standard variational-calculus techniques to optimally solve the instance~\eqref{vc:bal}. However, since our primary interest lies in the optimal objective value rather than the specific optimal solution or its uniqueness, a simpler ordinary differential equation (ODE)-based approach suffices in our case and is also adequate for handling the case of \rank.

\begin{theorem}\label{thm:bal-a}
The variational-calculus instance~\eqref{vc:bal} has an optimal objective value of $1/\sfe$.
\end{theorem}

Note that Lemma~\ref{lem:10-16-a} and Theorem~\ref{thm:bal-a} together lead to Theorem~\ref{thm:bal}.

\begin{proof}[Proof of Theorem~\ref{thm:bal-a}]
Let $u(t):=\int_0^t g(z) \sd z$ for $t \in [0,1]$. Thus, we have that $u(0)=0$, $u(1)=\int_0^1 g(z) \sd z$, and $u'(t)=g(t) \ge 0$. Observe that the constraint of Program~\eqref{vc:bal} with $t=1$ suggests that   
\begin{align*}
\int_0^1 g(z) \sbp{2-z} \sd z \le 1 \Rightarrow \int_0^1 g(z) \sd z \le \int_0^1 g(z) \sbp{2-z} \sd z \le 1  \Rightarrow  u(1) \le 1.
\end{align*}
As a result, $u(t) \in [0,1]$ when $t \in [0,1]$.  The objective of Program~\eqref{vc:bal} is 
\[
\int_0^1 g(t) (1-t) \sd t  =\int_0^1 u'(t) (1-t) \sd t=\int_0^1 u(t) \sd t.
\]
Additionally, the constraint of Program~\eqref{vc:bal} is 
\begin{align*}
 \int_0^t  g(z) \sbp{1-z+t} \sd z &= \int_0^t  u'(z) \sbp{1-z+t} \sd z=u(t)+ \int_0^t u(z) \sd z.
\end{align*}
Thus, Program~\eqref{vc:bal} can be stated alternatively as follows:
\begin{align} \label{vc:balb}
\max_{u: [0,1] \to [0,1]}  \int_0^1 u(t) \sd t:~~
 u(t)+  \int_0^t u(z) \sd z\le t, \forall t \in [0,1]; u(0)=0; u'(t) \ge 0, \forall t\in [0,1].
\end{align}

Let $v(t):= \int_0^t u(z) \sd z$ for $t \in [0,1]$. We can re-state Program~\eqref{vc:balb} as follows:
\begin{align}\label{vc:balc}
\max_{v: [0,1] \to [0,1]} & v(1):\\
& v(t)+v'(t)\le t, \forall t \in [0,1]; \nonumber \\
& v(0)=v'(0)=0; 0 \le v'(t) \le 1, v''(t) \ge 0, \forall t \in [0,1].\nonumber
\end{align}

Let $v^*(t)$ be the unique function satisfying $v(t)+v'(t)=t$ for all $t \in [0,1]$ with $v(0)=0$ such that $v^*(t)=\sfe^{-t}-(1-t)$. We can verify that $v^*(t)$ is feasible to Program~\eqref{vc:balc}. Now, we show that for any feasible solution $v(t)$ to Program~\eqref{vc:balc}, we have $v(t) \le v^*(t)$ for all $t \in [0,1]$. Let $h(t):=v^*(t)-v(t)$ with $h(0)=0$. Observe that 
\begin{align*}
v(t)+v'(t) \le t=v^*(t)+\frac{d v^*(t)}{dt}, \forall t \in [0,1] & \Rightarrow h'(t)+h(t) \ge 0, \forall t \in [0,1]; \\
& \Rightarrow \sbp{h(t) \cdot \sfe^t}' \ge 0,  \forall t \in [0,1],
\end{align*}
which implies that the function $h(t) \cdot \sfe^t$ is non-decreasing over $t \in [0,1]$. Since $h(t) \cdot \sfe^t=0$ when $t=0$, we claim $h(t) \cdot \sfe^t \ge 0$ over $t\in [0,1]$, which implies $h(t) \ge 0$ and thus, $v^*(t) \ge v(t)$ for all $t \in [0,1]$. Therefore, for any feasible solution $v(t)$ to Program~\eqref{vc:balc}, $v(1) \le v^*(1)=1/\sfe$. Since $v^*(t)$ is feasible as well, we establish our claim. 
\end{proof}

\section{Resolving the Policy-Revealing LPs of the Classical Secretary Problem via a Variational-Calculus Approach}\label{sec:sp}

\subsection{Formulating $\LP_S(n)$ as a Variational-Calculus  Instance}\label{sec:sp-a} 
Recall that $\LP_{S}(n)$ denotes the policy-revealing LP for the classical secretary problem as shown in \LP-\eqref{lp:sp}, which is parameterized with an integer $n \ge 1$ representing the total number of candidates.

\begin{lemma}\label{lem:10-16-a}
When $n \to \infty$, the policy-revealing  $\LP_S(n)$ in~\LP-\eqref{lp:sp} becomes equivalent to the variational-calculus instance~\eqref{vc:sp}.
\end{lemma}

\begin{proof}
For each $i \in [n]$, let $g(i/n)=i \cdot x_i \in [0,1]$, where $g(t)$ is a function with domain and range both on $[0,1]$. As a result, the objective of $\LP_S(n)$ when $n \to \infty$ can be approximated as
\[
\sum_{i \in [n]} x_i \cdot (i/n)=\frac{1}{n}\sum_{i \in [n]} g(i/n) \longrightarrow \int_0^1 g(t) \sd t, ~~\sbp{\mbox{when $n \to \infty$}}.
\]

As for the constraint of $\LP_S(n)$  in~\LP-\eqref{lp:sp},
\begin{align*}
 & x_i  \cdot i  \le 1-\sum_{1 \le \ell<i} x_\ell,  \forall i \in [n]  \Leftrightarrow
 x_i  \cdot i  \le 1-\sum_{1 \le \ell<i} x_\ell,   i=\lfloor t \cdot n \rfloor, \forall t \in [0,1].\\
  \Leftrightarrow&
g(\lfloor t \cdot n \rfloor/n) \le 1- \frac{1}{n}\sum_{1 \le \ell<\lfloor t \cdot n \rfloor} \frac{g(\ell/n)}{\ell/n}, \forall t \in [0,1]. \\
   \Leftrightarrow& g(t) \le 1-\int_0^t \frac{g(z)}{z} \sd z,~~  \forall t \in [0,1]~~\sbp{\mbox{when $n \to \infty$}}.
 \end{align*}
 The constraint $g(t) \ge 0$ for all $t \in [0,1]$ is valid due to the definition of $g(t)$. Thus, we establish our claim.
\end{proof}

\subsection{Solving the Variational-Calculus Instance~\eqref{vc:sp} Optimally}\label{sec:29-c}
Recall that, in our case, we are concerned with not only the optimal objective value but also  the corresponding optimal solution and its uniqueness. In this section, we apply the classical Lagrange-Multipliers method to first derive the necessary conditions that any optimizer should satisfy and then determine the unique maximizer from those conditions.

Let $u(t):=\int_0^t g(z)/z~\sd z$, and $\du(t)=g(t)/t$.\footnote{Note that by the well-known Newton-Leibniz formula, we  have that \( u(t) \) is \emph{absolutely continuous} over \( [0,1] \), as we can safely assume \( g(z)/z \) is Lebesgue integrable over \( [0,1] \). Meanwhile, the equality \( \du(t) = g(t)/t \) actually holds almost everywhere on \( [0,1] \) and everywhere on \( [0,1] \) if \( g(t)/t \) is continuous over \( [0,1] \).} 
We can re-write the variational instance~\eqref{vc:sp} as follows:
\begin{align}\label{vc:sp-b}
\max_{u(t):t\in [0,1]} &~\int_0^1 \du(t) \cdot t~\sd t:~~u(t)+\du(t) \cdot t \le 1, \forall t \in [0,1]; u(0)=0, \du(t) \ge 0,\forall t\in [0,1].\footnotemark
\end{align}
\footnotetext{Note that for this variational instance, we omit implicit constraints on $u(t)$, such as (1) $u(t) \ge 0$ over $t \in [0,1]$, since $\du(t) \ge 0$ for all $t \in [0,1]$ and $u(0)=0$; and (2) $u(t) \le 1$ over $t \in [0,1]$, since $\du(t) \ge 0$ and $u(t) + \du(t) \cdot t \le 1$ for all $t \in [0,1]$.}

\begin{theorem}\label{thm:10-20-a}
There is a unique maximizer of $u^*$ to the variational instance~\eqref{vc:sp-b}, where
\[
u^*(t)=0, t \in [0,1/\sfe], u^*(t)=1-(1/\sfe)/t, t \in [1/\sfe,1].
\]
\end{theorem}

Observe that Lemma~\ref{lem:10-16-a} and Theorem~\ref{thm:10-20-a} together lead to the main Theorem~\ref{thm:sp}. By introducing auxiliary variables, we reformulate the inequality constraints in the instance~\eqref{vc:sp-b} as equalities, as shown below:
\begin{align}\label{eqn:7-25-a}
\max_{u, v, w} &~\int_0^1 w^2\cdot t~\sd t \\
& u(t)+w^2 (t) \cdot t+v^2(t)= 1, \forall t \in [0,1]. \label{eqn:7-25-b} \\
&\du(t)-w^2(t)=0, \forall t \in [0,1].  \label{eqn:7-25-c}\\
& u(0)=0.\label{eqn:7-25-d}
\end{align}
Now, we apply the classical Lagrange-Multipliers method to convert the constrained version above to an {unconstrained} one:
\begin{align*}
\widehat{J}[u,v,w,\mu_1, \mu_2]=\int_0^1 \sd t~\cdot \bp{w^2\cdot t+\mu_1 \cdot \bb{u+w^2 \cdot t+v^2-1}+\mu_2 \cdot \bb{\du-w^2}},
\end{align*}
where $\mu_1(t)$ and $\mu_2 (t)$ are two multipliers associated with the two equality constraints. 

\begin{lemma}\label{lem:25-a}
Any optimizer of the  instance~\eqref{eqn:7-25-a}-\eqref{eqn:7-25-d} must satisfy the following conditions:
\begin{align}
 \dot{\mu}_2-\mu_1& \equiv 0,\label{eqn:7-25-2}\\
 v \cdot \mu_1& \equiv 0,\label{eqn:7-25-3}\\
 w \cdot \bb{\mu_2-t(1+\mu_1)}& \equiv 0.\label{eqn:7-25-4}
\end{align}
\end{lemma}
In the lemma above, we omit the variable \( t \) from each of the functions \( \mu_1 \), \( \mu_2 \), \( v \), and \( w \). Conditions~\eqref{eqn:7-25-2}-\eqref{eqn:7-25-4} imply that the corresponding functions are equal to zero pointwise over \( t \in [0,1] \). We leave the proof of Lemma~\ref{lem:25-a} to the next Section~\ref{sec:31-a}. Now, we show how Lemma~\ref{lem:25-a}  leads to Theorem~\ref{thm:10-20-a}.

\begin{proof}[Proof of Theorem~\ref{thm:10-20-a}]
We begin with the necessary conditions imposed on any optimizer (whether a minimizer or maximizer), as stated in Lemma~\ref{lem:25-a}. From Condition~\eqref{eqn:7-25-4}, we conclude that either \( w = 0 \) or \( \mu_2 = t(1 + \mu_1) \). Observe that for any interval \( [a, b] \) with \( 0 \leq a < b \leq 1 \) such that \( w(t) = 0 \) over \( t \in [a, b] \), we find that \( \du = 0 \) by Constraint~\eqref{eqn:7-25-c}. Thus, \( u(t) \) remains constant over \( [a, b] \).  
Consider any given range $[\hat{a},\hat{b}]$ with $0 \le \hat{a}<\hat{b} \le 1$, where $w\neq 0$, \ie $\du>0$ for all $t \in [\hat{a},\hat{b}]$.  Condition~\eqref{eqn:7-25-4} suggests that $\mu_2=t(1+\mu_1)$ for all $t \in [\hat{a},\hat{b}]$. By Condition~\eqref{eqn:7-25-2}, we have that $\dot{\mu}_2=\mu_1$. Therefore,
\[
\dot{\mu}_2=1+\mu_1+t \cdot \dot{\mu}_1=\mu_1 \Rightarrow
\dot{\mu}_1=-1/t \Rightarrow \mu_1=-\ln t+c,
\]
where $c$ is a constant. Observe that  $\mu_1$ has at most one point equal to zero over $[\hat{a},\hat{b}]$. By Constraint~\eqref{eqn:7-25-3}, we claim that $v=0$ everywhere in $[\hat{a},\hat{b}]$ except for at most one point. Note that by Constraint~\eqref{eqn:7-25-b}, $v=0$ suggests that 
\[
\du \cdot t+u=1.
\]
Wrapping up the analyses above, we claim that any optimizer $u$ to the instance~\eqref{vc:sp-b} must satisfy either $\du=0$ ($u$ takes a constant value) or $\du \cdot t+u=1$. Observe that $u(0)=0$ as stated in Constraint~\eqref{eqn:7-25-d}. Let $[a_1, b_1]$ be the first interval with $0< a_1 \le b_1 \le 1$ such that $u \equiv 0$ for all $0 \le t<a_1$ and (2) $\du \cdot t+u=1$ for all $t \in  [a_1, b_1]$. By the continuity of the function $u$, we have $u(a_1)=0$. As a result, we can solve that
\[
u(t)=1-a_1/t, \forall t \in  [a_1, b_1].
\]

Consider an optimizer $u_\mathbf{s}$ that is parameterized with a sequence $\mathbf{s}:=(a_1,b_1, a_2,b_2, \ldots, a_K, b_K)$ satisfying $0<a_1 < b_1 <  \cdots <a_K< {b}_K \le 1$, such that $\du \cdot t+u=1$ in each interval $[{a}_\ell, {b}_\ell]$ with $\ell \in [K]$, and $\du=0$ otherwise. By the continuity of the function $u_\mathbf{s}$, we have $u_\mathbf{s}(t)=1-{a}_1/t$ for all $t \in [a_1, b_1]$ and $u_\mathbf{s}(t)=1-({a}_1/{b}_1) \cdot {a}_2 /t$ for all  $t \in [a_2, b_2]$. More generally, 
\[
u_\mathbf{s}(t)=1-\prod_{1\le i <\ell} (a_i/b_i) \cdot a_\ell/t, \forall t \in [a_\ell, b_\ell], \forall \ell=1,2,\ldots,K.
\]

For the optimizer $u_{\mathbf{s}}$, the objective of Instance~\eqref{vc:sp-b} is equal to
\begin{align}
\int_0^1 \du_\mathbf{s} (t) \cdot t ~\sd t&=\sum_{\ell =1}^K \int_{{a}_\ell}^{{b}_\ell}\du_\mathbf{s} (t) \cdot t  ~\sd t=\sum_{\ell =1}^K \int_{{a}_\ell}^{{b}_\ell} (1-u_\mathbf{s}(t))~\sd t \nonumber \\
&=\sum_{\ell =1}^K \int_{{a}_\ell}^{{b}_\ell} \bP{\prod_{1\le i <\ell} (a_i/b_i) \cdot a_\ell/t }~\sd t=\sum_{\ell =1}^K \prod_{1\le i <\ell} (a_i/b_i) \cdot a_\ell \ln \sbp{b_\ell/a_\ell}:=g(\mathbf{s}). \label{ineq:26-a}
\end{align}

By Lemma~\ref{lem:26-a}, there exists a unique sequence $\mathbf{s}^*$ that maximizes  $g(\mathbf{s})$, which is $\mathbf{s}^*=(a^*_1=1/\sfe, b^*_1=1)$ with $K=1$. This corresponds to a unique maximizer $u_{\mathbf{s}^*}$ for the instance~\eqref{vc:sp-b}, which takes the form of $
u_{\mathbf{s}^*}(t)=0$ for all $t \in [0,1/\sfe]$, and  $u_{\mathbf{s}^*}(t)=1-(1/\sfe)/t$ for all  $t \in [1/\sfe,1]$.
\end{proof}

\begin{lemma}\label{lem:26-a}
There is a unique sequence $\mathbf{s}^*$ that maximizes $g(\mathbf{s})$ in~\eqref{ineq:26-a}, which is $\mathbf{s}^*=(a^*_1=1/\sfe, b^*_1=1)$.
\end{lemma}
The proof of Lemma~\ref{lem:26-a} can be found in Appendix~\ref{app:lem-s}.

\subsection{Proof of Lemma~\ref{lem:25-a}}\label{sec:31-a}
Consider non-linear variations on the variables $u(t)$, $v(t)$, and $w(t)$ satisfying boundary conditions as follows:
\begin{align}
&u(\ep,t)~\Big|_{\ep=0}=u(t), ~~\frac{\partial u(\ep,t)}{\partial \ep}~\Big|_{\ep=0}=\vph(t); \label{eqn:24-a}\\
&v(\ep,t)~\Big|_{\ep=0}=v(t), ~~\frac{\partial v(\ep,t)}{\partial \ep}~\Big|_{\ep=0}=\psi(t);\label{eqn:24-b}\\
&w(\ep,t)~\Big|_{\ep=0}=w(t), ~~\frac{\partial w(\ep,t)}{\partial \ep}~\Big|_{\ep=0}=\eta(t);\label{eqn:24-c}
\end{align}
where $\vph, \psi$ and $\eta$ are three independent variational directions imposed on $u$, $v$, and $w$, respectively.\footnote{Note that we skip introducing variations to $\mu_1(t)$ and $\mu_2(t)$ since  they only lead to the two equality constraints~\eqref{eqn:7-25-b} and~\eqref{eqn:7-25-c}; see {Chapter 7 of the book~\cite{olver2021calculus}}.}  Observe that $u(0)=0$ implies that ${\vph}(t)=0$ at $t=0$, \ie $\vph(0)=0$.  Define

\begin{align*}
h(\ep):=\int_0^1 \sd t~\cdot \bP{w^2(\ep,t)\cdot t+\mu_1(t) \cdot \bb{u(\ep,t)+w^2(\ep,t) \cdot t+v^2(\ep,t)-1}+\mu_2 (t)\cdot \bb{\frac{\partial u(\ep,t)}{\partial t} -w^2(\ep,t)}}
\end{align*}

Observe that 
\begin{align*}
\frac{\partial \sbp{w^2(\ep,t)\cdot t}}{\partial \ep}=2 w(\ep,t) \cdot t \cdot \frac{\partial w(\ep,t)}{\partial \ep},
\end{align*}
which implies that
\begin{align*}
\frac{\partial \sbp{w^2(\ep,t)\cdot t}}{\partial \ep}~\Big|_{\ep=0}=2 w(\ep,t) \cdot t \cdot \frac{\partial w(\ep,t)}{\partial \ep}~\Big|_{\ep=0}=2 \cdot w(t) \cdot t \cdot \eta(t),
\end{align*}
which is due to the boundary conditions imposed on the variation $w(\ep,t)$ of $w(t)$, as suggested in~\eqref{eqn:24-c}. 
As a result, 
\begin{align*}
h'(\ep)~\Big|_{\ep=0}=\int_0^1 \sd t & \left( 2 w(t) \cdot t \cdot \eta(t)+ \mu_1 (t)\bb{\vph(t)+2w(t) \cdot t \cdot \eta(t)+2v(t) \cdot \psi(t)} \right. \\
&\left. +\mu_2(t)\bb{\dot{\vph}(t)-2w(t)\cdot \eta(t)}\right).
\end{align*}

For ease of notation, we omit the variable \( t \) in each function above from now on. By default, the terms \( u \), \( v \), and \( w \) refer to the single-variable functions \( u(t) \), \( v(t) \), and \( w(t) \), respectively, rather than their variations. As a result,
\begin{align*}
h'(0)&=\int_0^1 \sd t~\cdot \bp{2 w \cdot t \cdot \eta+\mu_1 \bb{\vph+2w \cdot t \cdot \eta+2v \cdot \psi}+\mu_2\bb{\dot{\vph}-2w\cdot \eta}}\\
&=\mu_2 (1) \vph(1)-\mu_2 (0) \vph(0)+\int_0^1 \sd t~\cdot \bp{2 w t \cdot \eta+\mu_1 \bb{\vph+2wt \cdot \eta+2v \cdot \psi}-\bb{\dot{\mu}_2\vph+\mu_2\cdot 2w\eta}}\\
&=\mu_2 (1) \vph(1)+\int_0^1 \sd t~\cdot \bp{ \vph \cdot \bb{\mu_1- \dot{\mu}_2}+ 2\psi \cdot (v \mu_1)+\eta \cdot 2 w\bb{t(1+\mu_1)-\mu_2} },~~\sbp{\mbox{since $\vph(0)=0$}}.
\end{align*}

A necessary condition that any optimizer of the variance instance~\eqref{eqn:7-25-a} must satisfy is that \( h'(0) = 0 \) for any variations imposed on \( u \), \( v \), and \( w \), subject to the boundary conditions shown in~\eqref{eqn:24-a}-\eqref{eqn:24-c}. Note that \( \vph \), \( \psi \), and \( \eta \) represent three independent directions. By the Fundamental Lemma of the Calculus of Variations~\cite{lanczos2012variational,gelfand2000calculus}, we establish our claim.

\section{Conclusion and Future Directions}
In this paper, we propose a variational approach as an alternative for solving factor-revealing and policy-revealing LPs. Compared with existing methods, our approach is distinguished by its generality: once we reformulate the LP in the limit, as its size approaches infinity, as a variational instance, we can then invoke methods from the extensive toolkit of techniques developed for variational calculus to solve it.

Our work opens several new directions. One promising avenue is to apply the proposed approach to \emph{analytically} solve the factor-revealing LPs in~\cite{mahdian2011online}, which were introduced to analyze \rank for online matching under random arrival order. We expect to reformulate that factor-revealing LP as a \emph{multi-dimensional} variational instance. Fortunately, as mentioned, variational calculus is a well-established field with a vast array of techniques developed for various complex forms of variational instances, including multi-dimensional cases (see, e.g., Chapter 11 of~\cite{olver2021calculus}). We are confident that the variational approach can successfully solve the problem.

\newpage
\bibliographystyle{alpha} 
\bibliography{EC_21}
 
\clearpage
\appendix

\section{Proof of Theorem~\ref{thm:re-bal}} \label{app:bal}
The algorithm \bal was first introduced by~\cite{kalyanasundaram2000optimal} for the Online $b$-Matching problem, which proves that $\bal$ achieves \anhai{optimal} competitiveness of $1-1/\sfe$ when $b \to \infty$. For completeness, we reiterate it in Algorithm~\ref{alg:bal}.

\begin{algorithm}[th!] 
\caption{\bal for Online $b$-Matching~\cite{kalyanasundaram2000optimal}.}\label{alg:bal}
\DontPrintSemicolon
\tcc{\bluee{Input is a bipartite graph $G=(U,V,E)$. Each $u$ has a matching capacity of $b \in \mathbb{Z}_{\ge 1}$, and the vertices of $V$ are revealed sequentially in an adversarial order. Output is matching over $G$}.}
When an online vertex $v \in V$ arrives: Match $v$ to an available neighbor $u^*$ of $v$, if any, where $u^*$ has the largest remaining capacity so far.
\end{algorithm}

\xhdr{Viewing Online $b$-Matching as an Adwords Instance}. Recall that for Online $b$-Matching, we have an input graph $G=(U,V,E)$ such that each $u$ has a matching capacity of $b$. Throughout this section, we refer to vertices of $U$ as bidders and those of $V$ as queries. We can interpret Online $b$-Matching as an Adwords problem such that each bidder $u$ has a unit budget, while every bid (edge weight) is equal to $1/b \ll 1$. Under this context, $\bal$ always assigns an arriving query $v$ to a neighbor $u^*$ such that $u^*$ has the largest unused fraction of budget when $v$ arrives.

\begin{proof}[Proof of Theorem~\ref{thm:re-bal}]
Assume WLOG that $\opt=|U|=n$.\footnote{In the work~\cite{mehta2007adwords}, it has been explained in detail why we can assume WLOG that $\opt$ achieves a total reward of $n$ in the worst-case scenario, \ie $\opt$ exhausts the unit budget of every bidder.} In other words, we assume that in the optimal algorithm, each bidder exhausts its unit budget. Suppose we run $\bal$ and group all of the $n$ bidders into $N+1$ groups based on the fraction of \anhai{the} used budget in the end with $N \gg n$. Specifically, we say a bidder $u \in U$ falls  group $i \in [N]:=\{1,2,\ldots,N\}$ if it uses a fraction of budget between $(i-1)/N$ and $i/N$, \ie $\rho_u \in ((i-1)/N, i/N)$, where $\rho_u$ denotes the fraction of \anhai{the} budget spent by bidder $u$ in the end, and a bidder $u$ falls group $i=N+1$ if it uses up all budget with $\rho_u=1$. Here we \anhai{simply} assume WLOG that no bidder uses a fraction of budget exactly equal to $i/N$ for some $1 \le i<N$.\footnote{In fact, we can always make it by choosing an appropriate large enough value of $N$. For example, consider a case when a bidder $u$ has $\rho_u=k/N$ for some integer $k \in \{1,2,\ldots,N-1\}$. We can break it by resetting $N' \gets N+P$, where $P$ is any prime number with $P>N$. We can verify that $\rho'_u=k'/N'$, with $k'=k+k P/N$ being not an integer.} Let $\alp_i$ denote the number of bidders falling in group $i \in [N+1]$. Thus, $\sum_{i \in [N+1]} \alp_i=n$. The performance of $\bal$ (the total rewards) can be expressed  as
\begin{align}\label{eqn:9-20-c}
\bal &=\sum_{i \in [N]} \alp_i \cdot (i/N)+\alp_{N+1}-\ep \nonumber\\ 
&=\sum_{i \in [N]} \alp_i \cdot (i/N)+n-\sum_{i \in [N]} \alp_i -\ep=n-\bb{\sum_{i \in [N]}\alp_i \cdot \bp{1-\frac{i}{N}}+\ep},
\end{align}
where $\ep \in [0, n/N]$ is to offset the part of rewards overcounted for \bal: Since each bidder in whatever group may have been overcounted by at most $1/N$ rewards, the total of overcounting rewards should be no more than $n/N$. Suppose we split the whole unit budget of each bidder into $N$ slabs such that each slab $j \in [N]$ represents the fraction interval between $(j-1)/N$ and $j/N$. We assume each bidder first spends slab $j=1$, followed by $j=2,3,\ldots,N$. Let $\beta_j$ denote  the total amount of money spent by all the  $n$ bidders in the slab $j$ for each $j \in [N]$. As a result,
\begin{align}\label{ineq:9-20-a}
\beta_j=\sum_{i =j}^{N+1} \alp_i/N-\ep_j=\frac{1}{N} \bp{n-\sum_{ i<j} \alp_i}-\ep_j,
\end{align}
where $\ep_j \in [0,\alp_{j}/N]$ denotes the amount of overcounting rewards that is solely attributable to potential bidders in group $i=j$. A key observation, as shown in Lemma~\ref{lem:9-20-a}, is that

\begin{align}\label{ineq:9-20-b}
\sum_{j \le p} \beta_j \ge \sum_{i \le p} \alp_i, \forall p \in [N].
\end{align}
Substituting Equality~\eqref{ineq:9-20-a} into Inequality~\eqref{ineq:9-20-b}, we have for any $p \in [N]$, 
\begin{align}
&~\sum_{j \le p} \beta_j =\sum_{j \le p} \bP{\frac{1}{N} \bp{n-\sum_{i<j} \alp_i}-\ep_j}=
\frac{p \cdot n}{N}-\sum_{j \le p}\sum_{ i<j}\alp_j/N-\sum_{j \le p} \ep_j \ge \sum_{i \le p} \alp_i \\
\Leftrightarrow&~\frac{p \cdot n}{N}-\sum_{i \le p}\alp_i (p-i)/N-\sum_{j \le p} \ep_j \ge \sum_{i \le p} \alp_i \\
\Leftrightarrow &~ \sum_{i \le p} \alp_i \bp{1+\frac{p-i}{N}} \le\frac{p \cdot n}{N}-\sum_{j \le p} \ep_j. 
\end{align}
Now, wrapping up all the analysis above, the worst-case performance of $\bal$ can be formulated as the below LP:
\begin{align}\label{lp:9-20-a}
\min &~~n-\bb{\sum_{i \in [N]}\alp_i \cdot \bp{1-\frac{i}{N}}+\ep}\\
&  \sum_{i \le p} \alp_i \bp{1+\frac{p-i}{N}} \le\frac{p \cdot n}{N}-\sum_{j \le p} \ep_j, \forall p \in [N]; \nonumber\\
& \alp_i \in [0, n],\forall i \in [N]. \nonumber
\end{align}
Note that: (1) In the objective, the term $\ep \in [0,n/N]$ and in the constraint, the term $\sum_{j \le p}\ep_j \le \sum_{j \in [N]} \ep_j \le \sum_{j \in [N]} \alp_j/N=n/N$. Thus, by taking $N \gg n$, both terms can be ignored. (2) Since $\OPT=n$, the corresponding competitiveness achieved by \bal can be captured as the objective divided by $n$. Let $x_i:=\alp_i/n \in [0,1]$. We can characterize  the final competitiveness of $\bal$ using the below LP:
\begin{align}\label{lp:9-20-b}
\min &~~1-\bb{\sum_{i \in [N]} x_i \cdot \bp{1-\frac{i}{N}}}\\
&  \sum_{i \le p} x_i \bp{1+\frac{p-i}{N}} \le\frac{p }{N}, \forall p \in [N]; \nonumber\\
& \x_i \in [0, 1],\forall i \in [N]. \nonumber
\end{align}
Note that the above \LP-\eqref{lp:9-20-b} has an optimal value equal to $1-\val(\LP_B(N))$, where $\LP_B(N)$ is the factor-revealing LP for \bal in LP-\eqref{lp:bal}.  Thus, we establish our claim.
\end{proof}

\begin{lemma}\label{lem:9-20-a}
$\sum_{j \le p} \beta_j \ge \sum_{i \le p} \alp_i, \forall p \in [N]$.
\end{lemma}
\begin{proof}
Consider a given $p \in [N]$. Note that for the group $i=p$, denoted by $\cU_p$,  the optimal policy $\opt$ exhausts the unit budget for every bidder $u \in \cU_p$. Let $\cV_p$ denote the set of queries assigned by $\opt$ to any bidder  $u \in \cU_p$.  Consider a given $v \in \cV_p$. Note that at the time when $v$ arrives, $\opt$ makes the choice of assigning $v$ to some $u \in \cU_p$, where $u$ spends a fraction of budget less than $p/N$ (even by the end of $\bal$). Therefore, at the same time when $v$ arrives, $\bal$ will assign $v$ to some bidder $u$ that spends a \anhai{fraction} of budget no more than $p/N$, due to the nature of $\bal$. Consequently, the money obtained by $\opt$ from matching any queries in $\cV_p$ will flow to some slab $j$ with $ j\le p$ for $\bal$. More generally, the money obtained by $\opt$ from matching any queries in $\cup_{1\le i \le p}\cV_i$ will flow to some slab $j$ with $ j\le p$ for $\bal$. Observe that the total money obtained by $\opt$ from matching all queries in $\cup_{1\le i \le p}\cV_i$ is equal to  $\sum_{1\le i \le p}|\cU_i|=\sum_{1\le i \le p} \alp_i$ (since $\opt$ exhausts the unit budget for every bidder). This should be no more than the sum of money earned by $\bal$ over all slabs $j$ with $1\le j \le p$. Therefore, our claim is established.
\end{proof}

\section{Proof of Theorem~\ref{thm:rank}}\label{app:rank}
The classical algorithm \rank was first introduced by~\cite{kvv}, which proved that \rank achieves an optimal competitiveness of $1-1/\sfe$ for the Online Matching problem (under adversarial order).  For completeness, we state it in Algorithm~\ref{alg:rank}.

\begin{algorithm}[ht!] 
\caption{\rank for Online Matching~\cite{kvv}.}\label{alg:rank}
\DontPrintSemicolon
\tcc{\bluee{Input is a bipartite graph $G=(U,V,E)$, which is unknown to the algorithm. The vertices of $V$ are revealed sequentially in an adversarial order. Output is a matching $M$ of $G$}.}
\tbf{Offline Phase}: Choose a random permutation $\pi$ over $U:=[n]=\{1,2,\ldots,n\}$ uniformly at random.\;
\tbf{Online Phase}: When an online vertex $v \in V$ arrives: If there are available neighbors of $v$, add the edge $(u,v)$ to $M$, where $u$ has the smallest value under $\pi$ among these neighbors.
\end{algorithm}

The book~\cite{mehta2012online} (Chapter 3) offers a refined combinatorial analysis for \rank. The main idea is as follows: Assume, without loss of generality (WLOG), that there exists a perfect matching in the input graph $G$ such that $\OPT = n$. For each position $i \in [n]$, let $x_i \in [0,1]$ denote the probability that a match event occurs in \rank. Thus, the expected number of matches by \rank can be captured as $\sum_{i \in [n]} x_i$. As a result, the competitiveness achieved by \rank is equal to $\sum_{i \in [n]} x_i/n$, the exact objective of the factor-revealing $\LP_R(n)$ as shown in~\LP-\eqref{lp:rank}. The work~\cite{mehta2012online} justified the validity of the constraints in $\LP_R(n)$ for every $i \in [n]$. Therefore, the minimization program of~\LP-\eqref{lp:rank} offers a valid lower bound on the competitiveness achieved by \rank. We are interested in the asymptotic behavior of the optimal value of $\LP_R(n)$ as $n$ approaches infinity.

\subsection{Formulating $\LP_R(n)$ as a Variational-Calculus (VC) Instance} 

For each $i \in [n]$, let $g(i/n):=x_i$,
where $g(t)$ is a function with domain and range both on $[0,1]$. Observe that for the objective of $\LP_R(n)$ in LP-\eqref{lp:rank}, 
\[
\frac{1}{n} \sum_{i\in[n]}x_i =\frac{1}{n} \sum_{i\in[n]} g(i/n) \longrightarrow \int_0^1 g(t)~ \sd t, ~~\sbp{\mbox{when $n\to \infty$}}. 
\]

Additionally, for the constraint of $\LP_R(n)$ in LP-\eqref{lp:rank},
\begin{align*}
 &x_i+\frac{1}{n} \sum_{1 \le j \le i} x_j \ge 1,  \forall i \in [n]  \Leftrightarrow
g(i/n) +\frac{1}{n}\sum_{1 \le j \le i} g(j/n) \ge 1,  i=\lfloor t \cdot n \rfloor, \forall t \in [0,1].\\
\Leftrightarrow &  g(t)+\int_0^t g(z)\sd z \ge 1, \forall t \in [0,1]. ~~\sbp{\mbox{When $n \to \infty$}.}
 \end{align*}
 Wrapping up all the analysis above, we have that when $n \to \infty$, $\LP_R(n=\infty)$ becomes equivalent to the variational instance~\eqref{vc:rank}.

\subsection{A Simple Ordinary Differential Equation (ODE)-Based Approach to Program~\eqref{vc:rank}}
Similar to the case of \bal, we can apply a simple Ordinary Differential Equation (ODE)-based approach to solving Program~\eqref{vc:rank} optimally. Let $u(t):=\int_0^t g(z) \sd z$ for each $t \in [0,1]$. Thus, Program~\eqref{vc:rank} can be re-stated as:
\begin{align}\label{eqn:8-17-b}
\min u(1):~~u'(t)+u(t) \ge 1, \forall t\in [0,1]; u(0)=0; 0\le u'(t)\le 1, \forall t \in [0,1].
\end{align}

\begin{theorem}\label{thm:rank-a}
An optimal solution to Program~\eqref{eqn:8-17-b} is that  $u^*(t)=1-\sfe^{-t}$ for $t \in [0,1]$.
\end{theorem}
Note that Theorem~\ref{thm:rank-a} suggests that the optimal value of $\LP_R(n=\infty)$ should be equal to $u^*(1)=1-1/\sfe$, which establishes Theorem~\ref{thm:rank}.

\begin{proof}
Let $u^*(t)=1-\sfe^{-t}$ with $t\in [0,1]$ be the unique solution that satisfies the first constraint of Program~\eqref{eqn:8-17-b} with equality, i.e., $u'(t)+u(t) =1$ for all $t \in [0,1]$, and the initial condition $u(0)=0$. We can verify that $u^*(t)$ also satisfies the last constraint, making it feasible for Program~\eqref{eqn:8-17-b}. Now, we will show that for any feasible solution $u(t)$ to Program~\eqref{eqn:8-17-b}, $u(t) \geq u^*(t)$ for all $t \in [0,1]$. This implies that $u(1) \geq u^*(1)=1-1/\sfe$.

Let $h(t):=u(t)-u^*(t)$. Due to the feasibility of $u$ and $u^*$, we have $h(0)=0$ and $h'(t)+h(t) \geq 0$ for all $t \in [0,1]$. Therefore, $(h(t) \cdot \sfe^{t})' \geq 0$ for all $t \in [0,1]$.  This means that $f(t):=h(t) \cdot \sfe^t$ is non-decreasing over $t \in [0,1].$ Since $f(0)=0$, it follows that $f(t) \geq 0$ for all $t \in [0,1]$.  Hence, $h(t)=u(t)-u^*(t) \geq 0$ for all $t \in [0,1]$, implying that $u(t) \geq u^*(t)$ throughout $t \in [0,1]$.
\end{proof}

\section{Proof of Lemma~\ref{lem:26-a}}\label{app:lem-s}
\begin{proof}
Consider the scenario $K=1$. In this case, $\mathbf{s}=(a_1, b_1)$ with $0<a_1 < b_1 \le 1$, and 
\[
g(\mathbf{s})=a_1 \ln(b_1/a_1).
\]
We can verify that $g(\mathbf{s})$ gets maximized at $\mathbf{s}=\mathbf{s}^*$ with $b_1=1$ and $a_1=1/\sfe$. 

Now, we focus on the case $K=2$. In this case, we need to solve the following maximization program:
\[
\max g(\mathbf{s})=a_1 \ln(b_1/a_1)+(a_1/b_1) a_2 \ln(b_2/a_2):~~ 0<a_1 <b_1< a_2 < b_2 \le 1.
\]
Clearly, for the maximizer of $\mathbf{s}^*=(a^*_1, b^*_1, a^*_2, b^*_2)$,  we have $b^*_2=1$. Thus, the target function becomes
\begin{align}\label{eqn:27-a}
 g(\mathbf{s})=a_1 \ln(b_1/a_1)+(a_1/b_1) a_2 \ln (1/a_2).
\end{align}
Suppose $a_1$ and $b_1$ are fixed, and now we determine a maximizer of $a_2$ subject to $b_1 <a_2< 1$. Consider the following two cases.

\xhdr{Case 1}. $b_1 < 1/\sfe$. In this case, we claim that $a_2^*=1/\sfe$. As a result, the target function $g(\mathbf{s})$  becomes
\[
g(\mathbf{s})=a_1 \ln(b_1/a_1)+(a_1/b_1) \cdot (1/\sfe)=a_1 \cdot \bp{\ln b_1+ \frac{1}{ \sfe \cdot b_1} }+a_1 \cdot \ln (1/a_1).
\]
For any given $a_1>0$, we can verify that the function $\ln b_1+ \frac{1}{ \sfe \cdot b_1}$ is decreasing when $b_1 \in (a_1, 1/\sfe)$. Note that when $b_1=a_1$, $g(\mathbf{s})=1/\sfe$. As a result, we claim that $g(\mathbf{s})<1/\sfe$.

\xhdr{Case 2}. $b_1  \ge 1/\sfe$. In this case, we claim that the target function $g(\mathbf{s})$ in~\eqref{eqn:27-a} is a strictly decreasing function of $a_2 \in (b_1,1)$ given $a_1$ and $b_1$ being fixed. Note that when $a_2=b_1$, $g(\mathbf{s})=a_1\ln(1/a_1)$. Thus,
\[
g(\mathbf{s})<a_1\ln(1/a_1) \le 1/\sfe.
\]
Wrapping up the analysis in the two cases above, we find that $g^*_2(\mathbf{s})<1/\sfe=g^*_1(\mathbf{s})$, where $g^*_K(\mathbf{s})$ with $K \in \{1,2\}$ denotes the optimal value of $g(\mathbf{s})$ when $\mathbf{s}$ is restricted to have $2K$ disjoint points over $(0,1]$, representing $K$ disjoint intervals. By applying a similar analysis, we can prove that for any $K \ge 2$, the corresponding optimal values satisfy $g^*_{K+1}(\mathbf{s})\le g^*_{K}(\mathbf{s})$.  Thus, we establish the claim.
\end{proof}

\end{document}